\documentclass{amsart}
\numberwithin{equation}{section}
\usepackage[utf8]{inputenc}
\usepackage[T1]{fontenc}
\usepackage{lmodern}
\usepackage[margin=1.2in]{geometry}
\usepackage{setspace}
\usepackage[usenames, dvipsnames]{xcolor}
\usepackage{graphicx}
\usepackage{mathtools}
\usepackage{amssymb}
\usepackage{amsthm}
\usepackage{xspace}
\usepackage{tikz-cd}
\usepackage{tkz-euclide}
\usepackage{tkz-graph}
\usetikzlibrary{
  knots,
  hobby,
  decorations.pathreplacing,
  shapes.geometric,
  calc
}
\usepackage{bbm}
\usepackage{slashed}
\usepackage[backref=page, bookmarks=false]{hyperref}
\usepackage[capitalize, noabbrev]{cleveref}
\newcounter{listcounter}
\newcounter{listcounter2}
\usepackage{autonum}

\newcommand{\Z}{\mathbb Z}

\newcommand{\R}{\mathbb R}

\newcommand{\mcal}[1]{\mathcal{#1}}

\newsavebox{\pullback}
\sbox\pullback{%
\begin{tikzpicture}%
\draw (0,0) -- (1ex,0ex);%
\draw (1ex,0ex) -- (1ex,1ex);%
\end{tikzpicture}}

\DeclareMathOperator{\Hom}{Hom}

\newtheorem{lemma}[equation]{Lemma}
\newtheorem{corollary}[equation]{Corollary}
\newtheorem{proposition}[equation]{Proposition}

\newtheorem*{lemma*}{Lemma}

\theoremstyle{definition}
\newtheorem{example}[equation]{Example}

\theoremstyle{remark}

\crefname{thm}{Theorem}{Theorems}
\crefname{lem}{Lemma}{Lemmas}
\crefname{cor}{Corollary}{Corollaries}
\crefname{prop}{Proposition}{Propositions}
\crefname{ex}{Exercise}{Exercises}
\crefname{exm}{Example}{Examples}
\crefname{defn}{Definition}{Definitions}
\crefname{claim}{Claim}{Claims}
\crefname{rem}{Remark}{Remarks}
\crefname{fct}{Fact}{Facts}
\crefname{note}{Note}{Notes}

\DeclarePairedDelimiter\paren{(}{)}

\makeatletter
	\let\oldparen\paren
	\def\paren{\@ifstar{\oldparen}{\oldparen*}}
\makeatother

\usepackage{microtype}
\usepackage{hypcap}

\hypersetup{
 colorlinks,
 linkcolor={red!50!black},
 citecolor={green!50!black},
 urlcolor={blue!80!black}
}

\newcommand{\id}{\mathrm{id}}

\newcommand{\Spin}{\textrm{spin}}

\newcommand{\cO}{\mathcal{O}}

\newcommand{\free}{\mathrm{free}}
\newcommand{\tor}{\mathrm{tor}}

\newcommand{\mbbC}{\mathbb{C}}

\newcommand{\mbbH}{\mathbb{H}}

\newcommand{\mbbQ}{\mathbb{Q}}
\newcommand{\mbbR}{\mathbb{R}}

\newcommand{\mbbZ}{\mathbb{Z}}

\newcommand{\Spinc}{\Spin^c}
\newcommand{\HNtor}{H_1(N)_\tor}

\newcommand{\QmodZ}{\mbbQ/\mbbZ}
\newcommand{\Poincare}{Poincar\'e\ }
\newcommand{\Ann}{\mathrm{Ann}}
\newcommand{\lk}{\mathrm{lk}}
\newcommand{\CS}{\mathrm{CS}}
\newcommand{\tc}{\mathrm{tc}}
\newcommand{\dvol}{\mathrm{dvol}}
\newcommand{\cH}{\mathcal{H}}
\newcommand{\LambdaR}{\Lambda_\R}
\newcommand{\LambdaRplus}{\LambdaR^+}
\newcommand{\LambdaRminus}{\LambdaR^-}
\newcommand{\diag}{\mathrm{diag}}
\newcommand{\TO}{\mathrm{TO}}
\newcommand{\cD}{\mcal{D}}
\newcommand{\cA}{\mcal{A}}
\newcommand{\cB}{\mcal{B}}
\newcommand{\pitchforkrel}{\pitchfork^{\mathrm{rel}}}
\newcommand{\cK}{\mcal{K}}

\title{From bordisms of three-manifolds to domain walls between topological orders}
\author{Yu Leon Liu}
\address{Department of Mathematics, Harvard University, 1 Oxford St, Cambridge, MA 02139}
\email{yuleonliu@math.harvard.edu}
\author{Dalton A R Sakthivadivel}
\address{Department of Mathematics, CUNY Graduate Centre, 365 Fifth Avenue, New York, NY 10016}
\email{dsakthivadivel@gc.cuny.edu}


\begin{document}
\maketitle

\begin{abstract}
    We study a correspondence between spin three-manifolds and bosonic abelian topological orders. Let $N$ be a spin three-manifold. We can define a $(2+1)$-dimensional  topological order $\TO_N$ as follows: its anyons are the torsion elements in $H_1(N)$, the braiding of anyons is given by the linking form, and their topological spins are given by the quadratic refinement of the linking form obtained from the spin structure. Under this correspondence, a surgery presentation of $N$ gives rise to a classical Chern--Simons description of the associated topological order $\TO_N$. We then extend the correspondence to spin bordisms between three-manifolds, and domain walls between topological orders. In particular, we construct a domain wall $\cD_M$ between $\TO_N$ and $\TO_{N'}$, where $M$ is a spin bordism from $N$ to $N'$. This domain wall unfolds to a composition of a gapped boundary, obtained from anyon condensation, and a gapless Narain boundary CFT. 
\end{abstract}

\tableofcontents

\section{Introduction}\label{sec:intr}
Topological orders are gapped phases of matter that exhibit exotic topological properties \cite{wen1990topological}.\footnote{A quantum theory is \emph{gapped} if the gap above the groundstate in the spectrum of the Hamiltonian persists in the themodynamic limit. Otherwise it is \emph{gapless}.} In $2+1$ dimension, these systems host quasi-particle excitations, called anyons. Anyons have remarkable algebro-topological properties which we will now describe. 

The fusion of anyons defines a commutative product
\begin{equation}
    a \times b = \sum_c N^c_{ab} \, c,
\end{equation}
with $N^c_{ab} = N^c_{ba}$. A topological order is called \emph{abelian} if the fusion product is group-like, i.e. $N^c_{ab}$ is zero for all but one $c$, for which $N^c_{ab} = 1$. In this case the set of anyons forms a finite abelian group under the fusion product $\times$. 
Additionally, anyons braid with each other. There are two types of braidings. The (full-)braiding $S(a,b)$ between two anyons $a$ and $b$ captures the Aharanov--Bohm phase created by winding one anyon around another in a full circle. There is also the half-braiding $\theta_a \in U(1)$ obtained by winding $a$ in a half-circle around another anyon with the same label $a$. The half-braiding $\theta_a$ is also called the topological spin, as it measures the spin-statistics of $a$.
For example, $a$ is a boson (fermion) if $\theta_a$ is $1$ ($-1$). However in general its value may be any root of unity. In the abelian case, the full-braiding is in fact determined by the half-braiding 
\begin{equation}\label{eq:S-from-theta}
    S(a,b) = \frac{\theta_{a+b}}{\theta_a \theta_b}.
\end{equation} 

In this paper we focus on bosonic abelian topological orders. From now on whenever we say topological order we mean a bosonic abelian topological order in particular. Following \cite{belov2005classification, stirling2008abelian, kapustin2011topological, lee2018study}, these topological orders 
are classified by 
the topological spins and the fusion product of their anyons.\footnote{The mathematical structure of general bosonic topological orders are captured by unitary modular tensor categories (MTCs) \cite{kitaev2006anyons}. See \cite{etingof2016tensor} for an introduction to MTCs. In this language, abelian topological orders correspond to \emph{pointed} unitary MTCs. 
See \cite[\S 2.3]{lee2018study} 
for the construction of MTCs from the $(D,q)$ data.} We can write these data additively as pairs $(D, q)$,\footnote{Some also requires a chiral central charge $c$ mod 24 that satisfies the Gauss--Milgram formula \cite{cano2014bulk}. We will not include such data.} where $D$ is a finite abelian group and $q$ is a quadratic form whose associated bilinear form 
\begin{equation}
    L(a,b) \coloneqq q(a+b) - q(a) - q(b)
\end{equation}
is non-degenerate.
Here $D$ is the group of anyons under multiplication by the fusion product; the full-braiding and the half-braiding are given by 
\begin{equation}
    S(a,b) = e^{2\pi i L(a,b)}, \quad \theta_a = e^{2 \pi i q(a)}.
\end{equation}

\begin{table}
    \centering
     \renewcommand{\arraystretch}{2}
       \begin{tabular}{c|c|c}
       Geometry    & Physics & Section\\
       \hline
       Spin three-manifold  
       & Abelian topological order  & \S\,\ref{sec:3-manifold-to-TO}\\ 
       Kirby diagram 
       & Toroidal Chern--Simons & 
       \S\,\ref{sec:Dehn-surgery-and-CS}\\ 
        Simply connected four-manifold
        & Narain boundary CFT 
        &
        \S\,\ref{sec:geometry-of-gapless-boundaries}\\
        Wall surgery & Anyon condensation
        & \S\,\ref{sec:geo-of-anyon-cond}\\
        Bordisms & Domain walls & \S\,\ref{sec:general-bordism}
    \end{tabular}
    \vspace{1em}
    \caption{}\label{tab:correspondence-table}
\end{table}
In this paper we associate topological orders to spin three-manifolds, and domain walls to spin bordisms. Note that unlike much of the literature on topological field theory, we are {\it not} assigning invariants from those theories to three-manifolds; rather, we are assigning a theory itself to the manifold. This establishes a correspondence between the geometry of three- and four-manifolds and the algebraic theory of topological orders. We summarize the correspondence in \cref{tab:correspondence-table}. To demonstrate the various aspects of this correspondence, we also provide two extended examples: $U(1)$ level $n$ Chern-Simons theory (Examples \ref{ex:U1-level-n-1}, \ref{ex:U1-level-n-2}, \ref{ex:U1-level-n-3}) and the toric code (Examples \ref{ex:toric-code-1}, \ref{ex:toric-code-2}, \ref{ex:toric-code-3}, \ref{ex:toric-code-4}).

% Section 2.
In \cref{sec:3-manifold-to-TO} we  construct a topological order $\TO_N$ using a spin three-manifold $N$. 
The finite abelian group $D$ is $H_1(N)_\tor$, i.e. the torsion subgroup of $H_1(N)$.\footnote{$H_*$ with no coefficients written is to be understood as integral homology.}
The braiding comes from the symmetric non-degenerate linking form 
\begin{equation}
    L \colon \HNtor \otimes \HNtor \to \QmodZ.
\end{equation}
Furthermore, by \cite{lannes1974signature, massuyeau2003spin} the spin structure on $N$ induces a quadratic refinement $q$ of the linking form.\footnote{By Wu's formula, $w_2 = \mathrm{Sq}^1(w_1)$, so any oriented three-manifold is spinnable.} The pair $(\HNtor, q)$ defines a topological order $\TO_N$.\footnote{We give a physical interpretation to this assignment in an upcoming paper.}

% Kirby diagram => CS description
Interestingly, surgery constructions of three-manifolds give classical descriptions of a topological order in the form of toroidal Chern--Simons theories. On one hand, surgeries are the standard way to construct three-manifold. A well-known result of Lickorish \cite{lickorish1962representation} and Wallace \cite{wallace1960modifications} is that every closed, connected, orientable three-manifold arises from performing integral Dehn surgery on a link in the three-sphere. 
Furthermore, Kirby diagrams with entirely even framing numbers give rise to spin three-manifolds.
On the other hand, toroidal (i.e. $U(1)^n$) Chern--Simons theories \cite{belov2005classification} are the standard way to construct abelian topological orders. They take the input of a non-degenerate even symmetric bilinear form, which is called the $K$-matrix.

After establishing necessary results about surgery in \cref{sec:surgeries-3-manifold}, in \cref{sec:Dehn-surgery-and-CS} 
we assign a toroidal Chern-Simons theory $\CS_{\lk}$ to an even Kirby diagram, where the $K$-matrix $\lk$ is obtained from the linking numbers between pairs of link components.
Furthermore, the quantum theory of $\CS_{\lk}$ is the topological order $\TO_N$, giving us a commutative diagram:
\begin{equation}\label{eq:Kirby-calculus-and-CS}
    \begin{tikzcd}
        \textrm{Even Kirby diagrams} \ar[rrr, "\textrm{Perform surgery}"] \ar[d, "\CS_{\lk}", swap] &&& \textrm{Spin three-manifolds} 
        \ar[d, "\TO_N"]
        \\ 
        \textrm{Toroidal Chern--Simons theories} \ar[rrr, "\textrm{Quantum theory}"] &&& 
        \textrm{Abelian topological order}
    \end{tikzcd}
\end{equation}
This presents a nice parallel between constructing three-manifolds from Kirby diagrams, and quantizing classical Chern--Simons Lagrangians.

Now we move onto assigning a 
domain wall $\cD_M$, a two-dimensional interface between two topological orders $\TO_N$ and $\TO_{N'}$, to 
a spin bordism $M$ from $N$ to $N'$. Using the folding trick (\cref{subsec:folding-trick}), it suffices to consider the case where $N'$ is empty, i.e. $M$ is a spin four-manifold bounding $N$, and $\cD_M$ is a boundary of $\TO_N$. Note that this is somewhat counterintuitive, as $M$ is a four-manifold whose boundary is $N$ whereas $\cD_M$ is a two-dimesional boundary of $\TO_N$. This comes from the fact that we are compactifying a $6D$ theory on $M$ and $N$ (see \cref{sec:outlook}).

Let us review the structure of boundaries for a topological order $(D,q)$ \cite{ji2019noninvertible, lin2023asymptotic, lin2023bootstrapping}. Given an ordinary $2D$ theory $\cB$, we can view it as a boundary of the trivial topological order $(D = 0, 0)$. 
Its partition function $Z_{\cB}(\tau)$, i.e. the value of $\cB$ on a torus, 
is a function of $\tau \in \mbbH/SL(2,\Z)$.\footnote{In general $Z_\cB(\tau)$ is not holomorphic in $\tau$.} 
In particular, $Z_{\cB}$ is invariant under modular transformations:
\begin{equation}
    \begin{aligned}
        Z_{\cB}(\tau + 1) = Z_{\cB}, \\
        Z_{\cB}(-1/\tau) = Z_{\cB}.
    \end{aligned}
\end{equation}

Now suppose $\cB$ is a boundary theory for a general topological order $(D,q)$. 
Note that $\cB$ is allowed to have different chiral central charges $(c_+, c_-)$.
In constrast to the above, we have a {\it twisted} partition functions $Z_{\cB}^a(\tau)$, labeled by an anyon $a \in D$, obtained from evaluating $\cB$ on a solid torus with an $a$ anyon line inserted through the center.
These twisted partition functions satisfy the modular \emph{covariant} condition:
\begin{equation}\label{eq:modular-covariance}
\begin{aligned}
   Z_{\cB}^a(\tau+1) &= e^{-2\pi i\frac{( c_+ - c_-)}{24}}\,e^{2\pi i q(a)}\, Z_{\cB}^a(\tau),\\
   Z_{\cB}^a(-1/\tau) &= \frac{1}{\sqrt{| D|}}\,
\sum_{b \in D}\, e^{2\pi i L(a,b)}\, Z_{\cB}^b(\tau).
\end{aligned}
\end{equation}

In this paper we study both gapless and gapped boundary theories of topological orders.
In the gapless case, we consider Narain boundary CFTs \cite{MR4317226, GeneralizedNarain}. These boundary CFTs generalize the standard Narain CFTs, toroidal compactifications of heteroic string theory \cite{narain1989new}. 
Given a simply connected spin Riemannian four-manifold $M$ with boundary $N$, in \cref{sec:geometry-of-gapless-boundaries}  we associate a Narain boundary CFT $\cB_M$ to $\TO_N$. 
These boundary CFTs require the data of a polarization of $H^2(M, \R)$ (\cref{subsec:Narain CFT}), which comes from the Hodge isomorphism (\cref{subsec:Hodge-theory}). 

The gapped case is studied in \cref{sec:geo-of-anyon-cond}.
Gapped boundaries and domain walls of topological orders are governed by the theory of anyon condensation \cite{kong2014anyon}, which we review in \cref{subsec:anyon-cond-and-gapped-boundaries}.
In anyon condensation, we create a condensed phase by identifying a set of mutually commuting bosons with the vacuum and screening out those anyons that braid non-trivially with that set. 
In our setting, this translates to the question of how to kill homology classes\textemdash a central problem in surgery theory.
Taking a method for killing middle-dimensional homology classes described in \cite{wall1962killing}, we perform a surgery 
whose effect on homology is precisely the condensation of anyons.

Having considered the gapless and gapped case, in \cref{sec:general-bordism} we assign a boundary theory $\cD_M$ to $\TO_N$, where $M$ is any spin Riemannian four-manifold bounding $N$. This boundary is the composite of a gapped boundary from anyon condensation with a Narain boundary CFT, tieing together the two cases above. Lastly, using the folding trick, we assign domain walls $\cD_M \colon \TO_N \to \TO_{N'}$ to spin bordisms $M$ between $N$ and $N'$ in \cref{subsec:folding-trick}.

One can contrast our construction with similar literature, for instance \cite{gukov2017fivebranes}, where they assign Narain boundaries and topological orders to bordisms between spin three-manifolds using similar data. Importantly, they focus on the case where the linking form is positive-definite and include non-abelian generalizations of topological orders, and do not discuss gapped boundaries\textemdash whereas this work extends to general bordisms with indefinite linking forms, and introduces a surgery procedure which condenses the anyons in the theory.

In \cref{sec:outlook} we discuss possible generalizations of our correspondence and put it in the context of 3D-3D correspondences.

\textbf{Acknowledgements.}
The authors are grateful to Salvatore Pace, Oliver Thaker, and Ryan Thorngren for helpful conversations. YLL would also like to thank Mike Hopkins for his guidance and encouragement. DARS thanks Dennis Sullivan for insightful comments.

YLL gratefully acknowledges the financial support provided by the Simons Collaboration on Global Categorical Symmetries. DARS acknowledges support from the Einstein Chair programme at the Graduate Centre of the City University of New York.

\section{From three-manifolds to topological orders}\label{sec:3-manifold-to-TO}

In this section we construct topological orders from spin three-manifolds. 
We will start with oriented manifolds. Let $N$ be a closed, connected, oriented three-manifold and $\HNtor$ be the torsion subgroup of $H_1(N)$. Then \Poincare duality gives a non-degenerate symmetric bilinear form $L \colon \HNtor \otimes \HNtor \to \QmodZ$ called the linking form.\footnote{First defined by Seifert. See \cite{van1938invariants, wall1963quadratic, wall1967classification}. For a geometric picture of how the linking form is obtained generally from the intersection of submanifolds representing torsion cycles, see \cite[\S 5]{dennis}.} It can be defined as follows: 
given $a,b \in \HNtor$ and two knots $\cK, \cK'$ representing them. Note that all knots in this paper are oriented. Since $ra = 0$ in $H_1(N)$ for some $r$, there exists a Seifert surface $C$ which bounds $r\cK$, i.e. $\partial C = r\cK$. Now perturb $C$ and $\cK'$ so that they intersect transversally. Then the linking form is given by
\begin{equation}\label{eq:L-N}
 L(a,b) \coloneqq \frac{C \pitchfork \cK'}{r} \mod \Z.   
\end{equation}

Using the oriented structure on $N$, we have obtained a finite abelian group $\HNtor$ and a non-degenerate pairing $L$. This is sufficient to define the braiding of anyons but not the topological spin, which needs a quadratic refinement of the linking form.  Recall that a quadratic form refining $L$ is a map $q \colon \HNtor \to \QmodZ$
such that 
\begin{equation}
  q(na) = n^2 q(a), \quad q(a+b) - q(a) - q(b) = L(a,b).
\end{equation}
We follow the construction in \cite{massuyeau2003spin}:
fix a spin structure $\sigma$ on $N$ and some $a \in D$. Let $\cK$ be a knot in $N$ representing $a$, and $\nu$ a tubular neighborhood of $\cK$. Recall that a parallel $\ell$ of $\cK$ is a knot living on $\partial \nu$ whose homology class is $a$, i.e., the natural map $H_1(N - \nu) \to H_1(N)$ takes $[\ell]$ to $a$. Furthermore, choosing an isotopy class of a parallel is equivalent to choosing a framing $\nu = D^2 \times S^1$ of the normal bundle of $\cK$, a fact that we will use repeatedly.

Now, using the framing associated to $\ell$, the spin structure $\sigma$ on $N$ induces a spin structure $\sigma_\ell$ on $\cK$.
Recall that there are two spin structures on $\cK = S^1$: 
the \emph{bounding} spin structure which bounds the disk $D^2$, and the \emph{non-bounding} spin structure which does not. One can use this fact to define a $\Z_2$-valued invariant on the set of isotopy classes of parallels:
\begin{equation}\label{bounding-non-bounding-eq}
    \sigma(\cK, \ell) = 
    \begin{cases}
$0$ & \text{$\sigma_\ell$ is bounding} \\
$1$ & \text{$\sigma_\ell$ is non-bounding}
    \end{cases}
\end{equation}
Note that isotopy classes of parallels form a non-empty $\Z$-torsor, where $\ell + k$ is a parallel obtained by twisting $\ell$ around the meridian $k$ times. We will denote such a twist by $\ell + km$, with $m$ being the meridian.
By \cite[Lemma 10]{massuyeau2003spin}, we have
\begin{equation}\label{eq:need-this-1}
  \sigma(\cK, \ell + km) \equiv \sigma(\cK, \ell) + k \mod 2.
\end{equation}

Now we define the quadratic form $q(a)$. Suppose that $ra = 0$, and consequently that $r\cK$ is null-homologous; then there exists a Seifert surface $C$ whose boundary is $r\cK$. Choose a parallel of $\cK$ transverse to $C$ such that $\sigma(\cK, \ell) = 0$. Then we have
\begin{equation}\label{eq:definition-of-q}
    q(a) \coloneqq \frac{C \pitchfork \ell}{2r}. 
\end{equation}
Since ${C \pitchfork m} = r$, one can use \eqref{eq:need-this-1} to show that $q$ is well-defined. It is straightforward to check that $q$ is a quadratic form refining the linking form $L$.\footnote{
 Recall that the set of spin structures on $N$ form a torsor over $H_1(N, \Z_2)$. Given a spin structure $\sigma$ together with $s \in H_1(N, \Z_2)$, it is clear from the construction that the quadratic refinement $q'$ associated to $\sigma + s$ is of the form 
    \begin{equation}
        q'(a) = q(a) + s(a).
    \end{equation}
}
We have obtained a finite abelian group $D = \HNtor$, together with a quadratic form $q$, from the spin three-manifold $N$. 
These data allow us to define a topological order $\TO_N \coloneqq (\HNtor, q)$.

We end this section with two examples. Note in the following example $L(n,k)$ indicates a lens space\textemdash not to be confused with the linking form $L(a, b)$.
\begin{example}[Lens spaces part I]\label{ex:U1-level-n-1}
   Fix a positive integer $n$ and an integer $k$ coprime to $n$. 
    There is a free action of $\Z_n$ on $S^3$ seen as the unit sphere of $\mbbC^2$. The action of the generator is 
    \begin{equation}
        (z_1, z_2) \mapsto (e^{2\pi i/n}z_1, e^{2\pi i k/n} z_2).
    \end{equation}
    The lens space is the quotient $L(n,k) \coloneqq S^3/(\Z_n)$, and we have 
    \begin{equation}\label{eq:U1_n-D}
        H_1(L(n,k)) = \Z_n, \quad L(a,b) = \frac{k\,a\,b}{n} \mod \Z.
    \end{equation}
    Furthermore, the unique spin structure on $S^3$ descends to a spin structure on $L(n,k)$. 
    When $n$ is odd, the quadratic refinement is 
    \begin{equation}
        q(a) =  \frac{(k/2) \,a^2}{n} \mod \Z
    \end{equation}
    where $k/2$ is taken to be the unique element in $\Z_n$ with $2 \, (k/2) = k \mod n$.
    On the other hand, when $n$ is even, the associated quadratic refinement is 
    \begin{equation}\label{eq:U1_q}
        q(a) \coloneqq \frac{k\,a^2}{2n} \mod \Z.
    \end{equation}
    \eqref{eq:U1_n-D} and \eqref{eq:U1_q} define a topological order for any $L(n,k)$. 
    We will relate $L(n,1)$ to $U(1)$ level $n$ Chern--Simons theories in \cref{ex:U1-level-n-2}.
\end{example}

\begin{example}[Toric code part I]\label{ex:toric-code-1}
    The toric code \cite{kitaev2003fault}, i.e. the $\Z_2$ gauge theory \cite{dijkgraaf1990topological}, has anyon data 
    \begin{equation}
        D = \Z_2 \times \Z_2, \quad
        q(c) = \frac{a\, b}{2} \mod \Z
    \end{equation}
    where $c \coloneqq (a, b) \in \Z_2 \times \Z_2$. In particular, there are four anyons: three bosons $1 =(0,0)$, $e =(1,0)$, $m = (0,1)$, and a fermion $f =(1,1) = e \times m$.
    The braiding is given by
    \begin{equation}
        L((a_1, b_1), (a_2, b_2)) = 
        \frac{a_1\, b_2 + b_1\, a_2}{2}
        \mod \Z.
    \end{equation}
     We will construct a manifold whose topological order is the toric code in \cref{ex:toric-code-2}.
\end{example}

\section{Surgery on three-manifolds}\label{sec:surgeries-3-manifold}

\subsection{Surgery and spin structures}\label{subsec:surgery-and-spin}
Let $N$ be a three-manifold. The basic idea of (integral Dehn) surgery on three-manifolds is the following: we remove a framed embedding of the solid cylinder $D^2 \times S^1 \subset N$ and then glue back in $S^1 \times D^2$ to obtain a new oriented three-manifold:
\begin{equation}\label{eq:surgery-def}
    N' \coloneqq (N \setminus D^2 \times S^1) \bigcup_{S^1 \times S^1} (S^1 \times D^2).
\end{equation}
Note that the data of the embedded solid torus is 
equivalent to an oriented knot $\cK = \{0\} \times S^1$ accompanied by a choice of parallel $\ell = \{c\} \times S^1$, where $c$ is some point on the boundary of $D^2$. The manifold $N'$ obtained from surgery depends strictly on the isotopy class of $\ell$. 

Surgeries come with a canonical oriented bordism $M$, sometimes called the \emph{trace}, from $N$ to $N'$:
\begin{equation}\label{eq:trace-definition}
    M \coloneqq N \times I \bigcup_{\nu \times \{1\}} D^2 \times D^2.
\end{equation}
Geometrically,  we take the identity bordism $N \times I$, then glue on a two-handle $D^2 \times D^2$ by identifying the $D^2 \times S^1$ component of its boundary with $\nu \times \{1\}$.
More generally, we can perform surgery whenever we have a link $L$ together with a framing for each link component. This surgery comes with a canonical bordism from $N$ to the resulting manifold.

We are interested in surgeries that preserve spin structures. Equip $N$ with a spin structure $\sigma$ and $D^2 \times D^2$ with its unique spin structure.
We would like to determine if \eqref{eq:trace-definition} glues to a spin structure on $M$,\footnote{Note that a spin structure on $M$ induces a spin structure on $N'$.} which is equivalent to asking whether the spin structures on $N$ and $D^2 \times D^2$ restrict to the same spin structure on $\nu$.  
Since the spin structure on $D^2 \times D^2$ restricts to the trivial spin structure on $\nu$, this is in turn equivalent to whether the spin structure on $N$ restricts to the trivial spin structure on $\nu$, which is captured by $\sigma(\cK,\ell)$ defined in \eqref{bounding-non-bounding-eq}. To summarize, we have the following:
\begin{lemma}\label{lem:surgery-and-spin-structure}
    Let $\sigma$ be a spin structure on $N$. Given a knot $\cK$ in $N$ and a parallel $\ell$ of $\cK$, then \eqref{eq:surgery-def} defines a spin structure on $N'$ and \eqref{eq:trace-definition} defines a spin bordism from $N$ to $N'$ if and only if 
    \begin{equation}
        \sigma(\cK, \ell) = 0.
    \end{equation}
    More generally, given a link $L = \cK_1 \sqcup \cK_2 \sqcup \cdots \sqcup \cK_n$ in $N$ together with framings $\ell_i$ for each link component $\cK_i$, then performing $(L, \ell)$ surgery constructs a spin manifold $N'$ and a spin bordism $M \colon N \to N'$ if and only if 
    \begin{equation}
        \sigma(\cK_i, \ell_i) =0
    \end{equation}
    for every $i$.
\end{lemma}

\subsection{Surgeries and homology}
Now we discuss the effects of certain surgery moves on the first homology groups of three-manifolds.
Suppose we have an oriented three-manifold $N$ and a knot $\cK$ in $N$. Let $\nu$ be a tubular neighborhood of $\cK$ and $a$ be its homology class in $H_1(N)$. 
\begin{lemma}\label{lem:surgery-on-homology}
    {\cite[Lemma 1]{wall1962killing}}
    Let $m$ be the meridian class in $H_1(N-\nu)$. Then the canonical map $H_1(N- \nu) \to H_1(N)$ induces an isomorphism $H_1(N-\nu)/m \simeq H_1(N)$. Additionally, if $a$ is torsion in $H_1(N)$ then $m$ is torsion-free in $H_1(N-\nu)$.
\end{lemma}

Suppose $a$ is torsion. By \cref{lem:surgery-on-homology}
we have an extension:
\begin{equation}\label{eq:central-extension-helper}
    \Z \to H_1(N-\nu) \to H_1(N),
\end{equation}
where the generator of $\Z$ is the meridian $m$.
Let $H'_1(N-\nu)$ be the pullback 
$\HNtor \times_{H_1(N)} H_1(N-\nu)$. The central extension \eqref{eq:central-extension} pulls back to a central extension
\begin{equation}\label{eq:central-extension}
    \Z \to H'_1(N-\nu) \to \HNtor.
\end{equation}

\begin{proposition}
    Let $A$ be a torsion abelian group. The set of isomorphism classes of central extensions of $A$ by $\Z$ is equivalent to the set of homomorphisms $A \to \mbbQ/\Z$.
\end{proposition}
\begin{proof}
    There is a universal central extension 
\begin{equation}\label{eq:uni-extension}
     \Z \to \mbbQ \to \QmodZ.
\end{equation}
Given any homomorphism $f \colon A \to \Z$, we have a central extension $\hat{A} \coloneqq \mbbQ \times_{\mbbQ/\Z} A$ obtained by pulling back along $f$. Conversely, given a central extension $\hat{A}$ of $A$ by $\Z$, we can construct a map $f \colon A \to \QmodZ$ as follows. As $A$ is torsion, for every $a \in A$ we have $ra = 0$ for some $r$. Let $\tilde{a}$ be a lift of $a$ in $\hat{A}$. Then $r\tilde{a} = km$ for some $k$, where $m \in \hat{A}$ denotes the generator of the $\Z$ central subgroup. 
We now define $f(a) \coloneqq k/r \mod \Z$. One can readily check that this is independent of the lift, and that $f \longleftrightarrow \hat{A}$ are inverse constructions.\footnote{Let us also give a homological argument. Recall that $\mathrm{Ext}^1(A; \Z)$ classifies central extensions of $A$ by $\Z$. We claim that when $A$ is torsion, we have an isomorphism $\mathrm{Ext}^1(A; \Z) \simeq \Hom(A; \QmodZ)$. This follows from  applying $\mathrm{Ext}^*(A; -)$ to the short exact sequence \eqref{eq:uni-extension}.
}
\end{proof}
In \cite[Lemma 3]{wall1962killing},
Wall determined the map $\HNtor \to \QmodZ$ corresponding to \eqref{eq:central-extension}:
\begin{lemma}\label{lem:central-ext-classify}
    The central extension \eqref{eq:central-extension} is classified by the group homomorphism $L(a,-) \colon \HNtor \to \QmodZ$.
\end{lemma}
% Set $f$ equal to $B(a, -) : A \to \QmodZ$ for an arbitrary $a \in A$. In \cite{turaev1984cohomology} sufficient conditions are given for $B$ to be realized as the linking form on torsion of a spin manifold\textemdash namely, that $B$ is a non-degenerate symmetric bilinear form, $A$ is the torsion subgroup of a finitely generated abelian group, and that various compatibility conditions with spin structures arising from the Gauss--Milgram formula are satisfied. From this we immediately have
We get the following two useful corollaries:
\begin{corollary}\label{cor:Wall-prop}
     Suppose that the homology class of $\cK$ in $H_1(N)$ is torsion. There is an injection $H'_1(N - \nu) \hookrightarrow \HNtor \times \mbbQ$, whose image can be identified with $(b, k)$ such that $L(a,b) = k \mod \Z$.
\end{corollary}
\begin{corollary}\label{cor:torsion-element}
    The subgroup of torsion elements in $H_1(N-\nu)$ is in one-to-one correspondence with the annihilator subgroup $\Ann_{a}(\HNtor)$ of $\HNtor$, which consists of torsion elements $b \in \HNtor$ such that $L(a,b) = 0$. 
\end{corollary}
We also need the following useful lemma:
\begin{lemma}\label{lem:lem-2}
    If $L(a,a) = 0$, then 
    \begin{equation}
        q(a) = \sigma(\cK, \ell)/2 \mod \Z
    \end{equation}
    where $\ell$ is a parallel whose homology class $\tilde{a}$ is the unique torsion element in $H_1(N - \nu)$ lifting $a$. 
\end{lemma}
\begin{proof}
    This follows from the construction of $q(a)$.
\end{proof}

Now we choose a parallel $\ell$ of $\cK$. Its homology class $\tilde{a}$ in $H_1(N-\nu)$ is a lift of $a$. Let $N_\ell$
be the manifold obtained by doing surgery on $\ell$, and
$M$
 be the bordism from $N$ to $N'$. We have the following sequence of maps, the composites of which are homotopically equivalent:
\begin{equation}\label{eq:diamond-diagram}
    \begin{tikzcd}
        & N- \nu\ar[dl] \ar[dr]& \\ 
        N \ar[dr] & & N_\ell \ar[dl] \\ 
        & M
    \end{tikzcd}
\end{equation}
In that case, the following lemma holds:
% Letting $m, \tilde{a}$ denote the class of the meridian and parallel in $H_1(N-\nu)$, we ge the following:
\begin{lemma}
    Let $m$, $\tilde{a}$  denote the classes in $H_1(N-\nu)$ of the meridian and parallel respectively. Then the maps on $H_1$ of \eqref{eq:diamond-diagram} satisfy
\begin{equation}\label{eq:homology-diamond-diagram}
    \begin{tikzcd}
        & H_1(N- \nu) \ar[dl] \ar[dr]& \\ 
       H_1(N) \simeq H_1(N-\nu)/m \ar[dr] & & H_1(N_\ell) \simeq H_1(N-\nu)/\tilde{a} \ar[dl] \\
        & H_1(M) \simeq H_1(N-\nu)/(m, \tilde{a})
    \end{tikzcd}
\end{equation}
\end{lemma}
\begin{proof}
    The middle two identifications come from the discussion in \cite{wall1962killing}, whilst the bottommost follows from the fact that $M$ is homotopic to gluing a two-disc along $\cK$, from which we have
    \begin{equation}
        H_1(M) \simeq H_1(N)/a \simeq H_1(N-\nu)/(m, \tilde{a}).
    \end{equation}
\end{proof}

\section{Kirby diagrams and toroidal Chern--Simons theories}\label{sec:Dehn-surgery-and-CS}
In this section we relate Kirby diagrams, which are the standard blueprints for constructing three-manifolds, to toroidal Chern--Simons theories, which we view as classical descriptions of abelian topological orders.

\subsection{Kirby diagrams}\label{subsec:kirby-calculus}
\begin{spacing}{1.3}
A well-known result of Lickorish \cite{lickorish1962representation} and Wallace \cite{wallace1960modifications} is that every closed, connected, orientable three-manifold arises from performing integral Dehn surgery on a link in the three-sphere. This gives us a way to represent diffeomorphism classes of three-manifolds by links in $S^3$, which go under the name of \emph{Kirby diagrams}. We will review them in this subsection and describe how information about the linking form, and the trace $M$, is contained in the diagrams. (For a more in-depth review of how handlebody decompositions of four-manifolds are related to surgery presentations of their bounding three-manifolds, see \cite{mandelbaum1980four}.)

Given a knot $\cK$ in $S^3$, the framing on $S^3 = SU(2)$ gives a canonical framing of the knot, i.e.,  the isotopy class of a parallel for each knot, which is called the $0$-parallel or the $0$-slope. 
Let $L$ be a link in $S^3$. 
In order to perform surgery on $L$, we need to pick  a parallel class for each link component. In $S^3$, it follows that we simply need to assign an integer to each link component, which is called the framing number. 
A \emph{Kirby diagram} is a link diagram together with a framing number assigned to each component, and a \emph{surgery presentation} of a three-manifold $N$ is a Kirby diagram for which performing the surgery on $S^3$ prescribed by the diagram produces $N$.

Let $\sigma$ denote the spin structure on $S^3$ induced from its framing, and $\cK$ be a knot in $S^3$.
It is straightforward to see that the $k$-parallel $\ell_k$ satisfies 
\begin{equation}
    \sigma(\cK, \ell_k) = k \mod 2.
\end{equation}
It follows from \cref{lem:surgery-and-spin-structure} that a Kirby diagram with all framing numbers even produces a spin three-manifold.\footnote{Luckily, it is a theorem that every Kirby diagram can be reduced to one with entirely even framings using Kirby moves \cite[Remark 3]{kirby1978calculus}.} We will call such a Kirby diagram \emph{even}.

A surgery presentation of $N$ also comes with an oriented simply-connected four-manifold $M$ whose boundary is $N$:
\begin{equation}
    M \coloneqq D^4 \bigcup_{S^3} M',
\end{equation}
where $M' \colon S^3 \to N$ is the trace of the surgery.
Furthermore, by \cref{lem:surgery-and-spin-structure}, when the Kirby diagram is even, the manifold  $M$ is a spin manifold bounding $N$.

Since $M$ has a single zero-handle ($D^4$) and a two-handle for each link component, we see that $M$ is simply-connected with homology 
\begin{equation}
        H_*(M) = 
        \begin{cases}
            \Z  &* =0, \\
            \Z^{n} & * = 2, \\ 
            0& \mathrm{else.}
        \end{cases}
\end{equation}
Here $n$ is the number of components of $L$.
We write $e_i \in \Z^{n}$ for the basis element corresponding to the $i$-th component.  The intersection form $\lk \colon \Z^n \otimes \Z^n \to \Z$ can be read off from the Kirby diagram: in the $e_i$ basis, 
the diagonal elements $\lk_{ii}$ are equal to the framing numbers, and any off-diagonal element $\lk_{ij}$ is the linking number $\lk(\cK_i, \cK_j)$ between the $i$-th and $j$-th component of $L$. 
It follows that a Kirby diagram is even  if and only if $\lk$ is even.\footnote{Recall that an integral symmetric bilinear form $K$ is even if $K(v, v) \in 2\Z$ for every $v$.}
Now consider the relative homology sequence applied to $(M, N)$:\end{spacing}
\begin{center}\label{eq:surgery-presentation-LES}
    \begin{tikzcd} [column sep=0.3in, row sep=small] 
        * & H_*(N) & H_*(M) & H_*(M,N)\\
         4 & 0 & 0 & \mathbb{Z} \\
        3 & \mathbb{Z} & 0 & 0 \\
        2 & \ker \lk & \Lambda & \Lambda^* \\
        1 & \Lambda^*/\Lambda & 0 & 0 \\
        0 & 0 & \Z & \Z. \\
     \arrow[from=2-4, to=3-2, out=-20, in=160, "\simeq"']
     \arrow[from=4-2, to=4-3]
        \arrow[from=4-3, to=4-4, "f_\lk"]
        \arrow[from=4-4, to=5-2, out=-20, in=160]
        \arrow[from=6-3, to=6-4, "\simeq"]
    \end{tikzcd}
    \end{center}
Here we write $\Z^n$ abstractly as a lattice $\Lambda$, and view $\lk$ as an even integral symmetric bilinear form on $\Lambda$. By \Poincare duality, the relative homology $H_2(M,N)$ can be viewed as the dual of $\Lambda$, and the map 
$f_\lk \colon \Lambda \to \Lambda^*$ is the adjoint of $\lk$. It follows that the $H_1(N) \simeq \Lambda^*/\Lambda$.\footnote{The discussion can be understood geometrically: let $\alpha^*_i$ denote the dual element of $\alpha_i$ in $\Lambda^*$. Then the transgression map takes $\alpha^*_i$ to the homology class of the $i$-th meridian. Now counting punctures of Seifert surfaces with boundaries $\cK_i$, it can be shown that $H_1(N)$ with integral coefficients is generated by the $k$ meridians $m_1, \ldots, m_k$ subject to the set of relations $n_i m_i + \sum_{j \neq i} \lk(\cK_i, \cK_j) m_i = 0$ for each $i$. The reader may refer to \cite[\S 5]{gompf} for more details.
}
Furthermore, the linking form on $\HNtor$ can be defined as follow: given $a, b \in \HNtor$, 
let $\tilde{a}, \tilde{b}$ be lifts of $a$ and $b$ in $\Lambda^*$. Suppose that $ra = 0$, in which case $r\tilde{a}$ further lifts to a class $\tilde{a}'$ in $\Lambda$. We have 
\begin{equation}\label{eq:disc-L}
    L(a,b) = \frac{\tilde{b}(\tilde{a}')}{r} \mod \Z.
\end{equation}
Furthermore, suppose that $\lk$ is even, then the quadratic refinement $q$ is
\begin{equation}\label{eq:disc-q}
    q(a) = \frac{\lk(\tilde{a}', \tilde{a})}{2r^2}
    = \frac{\tilde{a}(\tilde{a}')}{2r} \mod \Z.
\end{equation}
Note this is independent of the choice of $\tilde{a}$ and $\tilde{a}'$ as $\lk$ is even. Note that by the argument in \cite[\S 5.3]{hopkins2005quadratic}, the descended form on the quotient is non-degenerate.\footnote{The passage from $(\Lambda, \lk)$ to $(\Lambda^*/\Lambda, q)$ 
also appears in number theory, where it is called the \emph{discriminant group} construction. See \cite{MR1625724}.}

\subsection{Toroidal Chern--Simons theory}\label{subsec:CS-theory}
We claim that the physical counterpart of surgery presentations of three-manifolds are toroidal Chern--Simons theories. Here we give a concise introduction to toroidal Chern--Simons theories, and refer the reader to \cite{belov2005classification} for more detail.

A toroidal Chern--Simons theory is given by a lattice $\Lambda \simeq \Z^n$ together with a non-degenerate even integral symmetric bilinear pairing $K \colon \Lambda \otimes \Lambda \to \Z$, typically called the $K$-matrix.\footnote{When $K$ is not even, the Chern--Simons theory defines a fermionic topological order; see \cite{belov2005classification}.} Recall that such a form is non-degenerate if $\ker K = 0$. Let $T = \Lambda \otimes \mbbR/\Z \simeq U(1)^n$ be the associated torus and $A_i$ be the $i$-th $U(1)$ gauge field.
The Lagrangian
\begin{equation}\label{eq:CS-Lagrangian}
    2 \pi i \int \, K^{ij} \, A_i \,  dA_j,
\end{equation}
defines a Chern--Simons theory $\CS_{\Lambda, K}$ with gauge group $T$. We will omit $\Lambda$ whenever there is no risk of confusion.

Now we study the quantum topological order associated to \eqref{eq:CS-Lagrangian}, following \cite{kapustin2011topological}. To start, there are topological Wilson lines labeled by the dual lattice $\Lambda^* = \Hom(\Lambda, \Z)$. Additionally, there are magnetic monopole operators labeled by $\Lambda$.
Furthermore, given $v \in \Lambda^*$, the magnetic monopole associated to $v$ has electric charge $f_K v \in \Lambda^*$, where $f_K \colon \Lambda \to \Lambda^*$ is the adjoint of $K$. 
For any Wilson line labeled by $l$, it follows that we 
can place a $v$ monopole at the junction between Wilson lines $l$ and $l+Kv$, thus identifying them.
Therefore the set of non-equivalent topological Wilson lines, i.e. the anyon lines, is precisely the discrminant group $D \coloneqq \Lambda^*/\Lambda$. 
The braiding $L$ and topological spin $q$ are given by 
\begin{equation}
    L(a,b) = \frac{\tilde{b}(\tilde{a}')}{r}, \quad q(a) = \frac{\tilde{a}(\tilde{a}')}{2r}.
\end{equation}
Here $a, b \in D$ are anyons with lifts $\tilde{a}, \tilde{b}$ in $\Lambda^*$, and $r\tilde{a} = f_K \tilde{a}'$ for some $r$ and $\tilde{a}' \in \Lambda$.

\subsection{From Kirby diagrams to toroidal Chern--Simons theories}\label{subsec:Kirby-to-CS}
Consider an even Kirby diagram with $d$ link components.
When the linking form $\lk$ on $\Lambda = \Z^d$ is non-degenerate, we can use $\lk$ to define toroidal Chern--Simons theory $\CS_\lk$. 
It follows from the discussion in the previous two subsections that the quantum theory of $\CS_\lk$ is precisely the topological order $\TO_N$ associated to $N$, where $N$ is the spin three-manifold obtained from the Kirby diagram. When $\lk$ is degenerate, the same discussion holds when we replace $\Lambda$ 
with $\Lambda/(\ker{\lk})$, which has a non-degenerate pairing.

Now we can show that every abelian topological order comes from a spin three-manifold:
by \cite{belov2005classification}, every abelian topological order $(D,q)$ comes from a toroidal Chern--Simons theory $\CS_K$, where $K$ is a non-degenerate bilinear form on some lattice $\Lambda$.
Suppose we build a Kirby diagram with $d = \mathrm{rank}(\Lambda)$ components, such that 
\begin{enumerate}
    \item the framing number for the $i$-th link component $\cK_i$ is $K_{ii}$. 
    \item the linking number between any two components $\cK_i$ and $\cK_j$ is $K_{ij}$.
\end{enumerate}
For example, we can simply make all components $\cK_i$ unknots that link with each other. 
By the discussions in \cref{subsec:kirby-calculus} and \cref{subsec:CS-theory}, we have the identification
\begin{equation}
    \CS_{K} = \TO_{N},
\end{equation}
where the spin three-manifold $N$ is obtained by performing surgery using the Kirby diagram. 

\begin{example}[Lens space and $U(1)_n$ part II]\label{ex:U1-level-n-2}
    Let us revisit the lens space example in \cref{ex:U1-level-n-1}. 
    Let $n$ be an even positive integer. The lens space $L(n,1)$ has a simple Kirby diagram:
    \begin{equation}\label{eq:kirby-Lens}
    \begin{tikzpicture}
        \draw[thick] (0,0) circle (0.5);
        \node at (0, 0.7) {$n$};
      \end{tikzpicture}
    \end{equation}
    The Chern--Simons theory asscociated to the Kirby diagram \eqref{eq:kirby-Lens} is precisely $U(1)$ level $n$, which has the following Lagrangian:
    \begin{equation}
      2\pi i  \int n\, A\, dA.
    \end{equation}
    Simiarly, $L(n, -1)$ lens spaces has a Kirby diagram consisting of an unknot with framing number $-n$. It corresponds to the $U(1)$ level $-n$ theory.
    \end{example}

    We also use the algorithm above to construct a spin three-manifold corresponding to toric code:
    \begin{example}[Toric code part II]\label{ex:toric-code-2}
    The toric code has a classical $U(1)^2$ Chern--Simons description by 
        \begin{equation}
            2 \pi i \int 2 (A_1 dA_2 + A_2 dA_1).
        \end{equation}
        We would like to find two knots that link twice with each other and assign the $0$ framing to both of them. The simplest such Kirby diagram is the following:
        
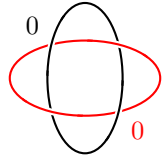
\begin{equation}\label{eq:Kirby-toric}
            \begin{tikzpicture}
            \begin{knot}[
            flip crossing=1,
            flip crossing=3,
            ]
            \strand[thick] (0,0) ellipse (0.5cm and 1cm);
            \strand[red, thick] (0,0) ellipse (1cm and 0.5cm);
            \node at (-0.7, 0.7) {$0$};
            \node[red] at (0.7, -0.7) {$0$};
            \end{knot}
            \end{tikzpicture}     
        \end{equation}
        This defines a spin three-manifold $N_{\tc}$ whose associated topological order is the toric code. 
        Note that the $e$ and $m$ anyons correspond to the meridians around the red and black unknots respectively.
        More generally, consider the even Kirby diagram of two unknots that link $n$ times around each other with $0$ framing for each component. The associated spin manifold corresponds to $\Z_n$ gauge theory.
    \end{example}

\section{Geometry of gapless boundaries}
\label{sec:geometry-of-gapless-boundaries}
In this section we show how simply-connected spin four-manifolds, such as the trace of an even Kirby digram, give rise to a family of gapless boundaries called the Narain boundary CFTs \cite{MR4317226, GeneralizedNarain}.

\subsection{Narain boundary CFTs}\label{subsec:Narain CFT}
In this subsection we follow the discussion in \cite[\S 2]{GeneralizedNarain}.
Fix a lattice $\Lambda$ and 
an even non-degenerate form $K$ on $\Lambda$ of signature $(p,q)$. 
The pair $(\Lambda, K)$ defines a toroidal Chern--Simons theory $\CS_K$.
Let $\LambdaRplus$ be a maximal positive definite subspace of $\LambdaR \coloneqq \Lambda \otimes \R$ and $\LambdaRminus$ be its orthogonal complements. We call a choice of $\LambdaRplus$ (equivalently $\LambdaRminus$) a \emph{polarization} of $\Lambda$, and denote a polarization by $m$.\footnote{Choose an identification $$\phi \colon (\LambdaR, K) \simeq (\R^{p|q}, \diag(\id_p, -\id_q)).$$ We see that one such $\LambdaRplus$ is $\R^{p|0}$. Furthermore, any other $\LambdaRplus$ is of the form $g \,\R^{p|0}$ for some $g \in O(p, q; \R)$. Since the stablizer of $\R^{p|0}$ is $O(p; \R) \times O(q; \R)$, we see that the space of polarization is the symmetric space
\begin{equation}
    O(p,q; \R)/O(p; \R)\times O(q;\R).
\end{equation}
In particular it only depends on the signature $p$ and $q$.\label{footnote:polarization-and-signature}}
Take some $v \in \Lambda$ and denote by $v^{\pm}$ its projection to $\LambdaR^{\pm}$.
Note that $K$ and $-K$ restrict to positive-definite definite norms on $\LambdaRplus$ and $\LambdaRminus$ respectively. We denote both norms by $|-|$.
The following identity holds: 
\begin{equation}\label{eq:def-of-K}
    K(v) \coloneqq K(v,v) = |v^+|^2 - |v^-|^2.
\end{equation} 
Additionally, 
we get a real-valued positive-definite norm $H$ on $\Lambda$ given by 
\begin{equation}\label{eq:def-of-H}
    H(v) \coloneqq |v^+|^2 + |v^-|^2.
\end{equation}

By \cite{GeneralizedNarain}, the data of $(\Lambda, K, m)$ defines a 2D CFT boundary $\cB_{\Lambda, K, m}$ of $\CS_K$ with central charge $(p,q)$  and $T_\Lambda$ symmetry, where $T_{\Lambda} \coloneqq \Hom(\Lambda, U(1))$ is the dual torus.\footnote{Conversely $\Lambda$ is the weight lattice of $T_{\Lambda}$.} Note that the data of $(\Lambda, K)$ is discrete while $m$ can be deformed continuously.
The primary local operators $\cO_v$ are labeled by $v \in \Lambda$, with energy and (worldsheet) momentum given by $H(v)$ and $K(v)$ respectively.
It follows that the partition function is given by 
\begin{equation}
    Z_{\Lambda, K, m}(\tau) = 
    \frac{1}{\eta(\tau)^p \bar{\eta}(\bar{\tau})^q} \sum_{v \in \Lambda} e^{\pi i y H(v) + \pi i x K(v)}.
\end{equation}
Here $\tau = x + i y$.
To capture the boundary data, we also need to consider the twisted sectors, which are labeled by the group of $\CS_K$ anyons $D = \Lambda^*/\Lambda$.
For $a \in D$, we have twisted primary operators $\cO_v$, where $v$ is now in the coset $\Lambda + a$, once again with energy $H(v)$ and momentum $K(v)$. The twisted partition function is 
\begin{equation}\label{eq:narain-partition-function}
    Z_{\Lambda, K,m}^a(\tau) = 
    \frac{1}{\eta(\tau)^p \bar{\eta}(\bar{\tau})^q} \sum_{v \in \Lambda + a} e^{\pi i y H(v) + \pi i x K(v)}.
\end{equation}
Note that 
\begin{equation}
    \Theta_{\Lambda, K, m}^a(\tau) \coloneqq \sum_{v \in \Lambda + a} e^{\pi i y H(v) + \pi i x K(v)}
\end{equation}
is the \emph{Siegel theta fuction} of $\Lambda$.
Here is how $\Theta_{\Lambda, K, m}^a$ behaves under modular transform:
\begin{equation}
    \begin{aligned}
        \Theta_{\Lambda, K,m}^a(\tau+1) &= 
        e^{2\pi i q(a)}\,
        \Theta_{\Lambda, K,m}^a(\tau, \bar{\tau}), \\ 
    \Theta_{\Lambda, K,m}^a(-1/\tau) &= 
    \frac{1}{\sqrt{|D|}} e^{-\frac{\pi i}{4}(p-q)}\tau^{p/2} \tau^{q/2}\sum_{b \in D}e^{2\pi i L(a,b)} \Theta_{\Lambda, K, m}(\tau).
\end{aligned}
\end{equation}
It follows that the Narain partition function  $Z_{\Lambda, K,m}^a(\tau)$ satisfies \eqref{eq:modular-covariance} and $\cB_{\Lambda, K, m}$ is a boundary of $\CS_K$.

\begin{example}[Narain CFTs]\label{ex:Narain-CFTs}
    Let $\Lambda = \Z^{2}$ and $H$ be the hyperbolic matrix. 
    Since $H$ is unimodular, $\CS_H$ is the trivial topological phase and its boundaries are two-dimensional modular invariant CFTs.
    The space of polarizations is parameterized by the positive real line, with $r$ corresponding to the polarization with $\LambdaRplus = \R(x_1 + r x_2)$. The associated Narain theory 
    $\cB_{\Z^2, H, r}$ 
    is the compact boson theory at radius $r$.

    More generally, take $\Lambda = \Z^{2n}$ and $K = H^{\oplus n}$. Since $H$ is unimodular we see that the Narain boundaries of $\CS_K$ are modular invariant CFTs with central charge $c = n$. Indeed we recover the original Narain conformal field theories, which are theories of $n$ compact bosons (see \cite[\S 2.1]{GeneralizedNarain}).
\end{example}

\begin{example}[Chiral bosons]\label{ex:chiral-bosons}
    Suppose $(\Lambda, K)$ is positive-definite. There is a unique choice of polarization with $\LambdaRplus = \LambdaR$ and $\LambdaRminus = 0$, hence a unique Narain boundary CFT.
    For example, let $\Lambda = \Z$ and $H = (n)$ where $n$ is an even positive integer. The unique Narain boundary of $U(1)$ level $n$ is precisely the chiral boson at radius $n$. The same discussion also holds for the negative-definite case.
\end{example}

\subsection{Hodge theory for manifolds with boundary}
\label{subsec:Hodge-theory}
Consider a closed oriented four-manifold $M$ and a Riemannian metric $g$ on $M$. There exists a Hodge star operator 
$\star  \colon \Omega^{i}(M) \to \Omega^{n-i}(M)$ such that for every $\alpha, \beta \in \Omega^{i}(M)$, we have 
\begin{equation}\label{eq:hodge-star-def}
    \alpha \wedge \star  \beta = (\alpha, \beta)\, \dvol.
\end{equation}
Here $(-, -) \colon \Omega^i(M) \otimes \Omega^i(M) \to \Omega^0(M)$ is given by the induced inner product on $i$-forms. Furthermore, $\dvol \in \Omega^4(M)$ is the volume form.

Now we restrict our Hodge star operator to harmonic forms.
Recall that a form $\alpha$ is \emph{harmonic} if $\nabla \alpha = 0$, i.e., it is a zero eigenfunction for the Laplacian operator $\nabla$. Note that when $M$ is compact, this is equivalent to $\alpha$ being closed ($d\alpha = 0$) and co-closed ($d\star  \alpha = 0$). 

Let $\cH^i(M) \subset \Omega^i(M)$ denote the subspace of harmonic $i$-forms on $M$. 
Since $\star $ takes closed to co-closed forms, it preserves the subspace of harmonic forms $\cH^*(M)$. Moreover, observe that $\star $ is an involution on $\cH^2(M)$, i.e., $\star ^2 = 1$. Therefore we can use $\star $ to split $\cH^2(M)$ into positive and negative eigenspaces. Let $\cH^{\pm} \subset \cH^2(M)$ denote the $\pm 1$ eigenspaces of $\star $. It follows from \eqref{eq:hodge-star-def} that $\cH^{\pm}$ are orthogonal positive and negative subspaces of $\cH$ with respect to the wedge pairing $\int_M - \wedge -$ on $\cH^2(M; \mbbR)$.

The main result of Hodge theory states that when $M$ is compact, $\cH^i(M)$ is isomorphic to $H^i(M; \mbbR)$, that is, every de Rham class has a unique harmonic representative. Therefore $\star $ defines a polarization on $H^2(M; \R) = H^2(M; \Z) \otimes \R$.
\begin{example}\label{ex:S2-times-S2}
    Let $M = S^2 \times S^2$. Here $H^2(M; \Z) = H_2(M; \Z) = \Z^2$, generated by $a = [S^2 \times 1]$ and $b = [1 \times S^2]$. The intersection form is the hyperbolic metric $H$.

    For $r \in \R_+$ let $g_r$ denote the round metric on $S^2$ with volume $r^2$.
    Consider the metric $g_{r_1} \times g_{r_2}$ on $S^2 \times S^2$. The corresponding Hodge star operator is given by $\star a = \frac{r_1}{r_2} b$. 
    It follows that the positive definite subspace is  $\R(a + \frac{r_1}{r_2} b)$. 
\end{example}
In our case, we are interested in the case where $(M, g)$ is an orientable Riemannian manifold with boundary, where the story is more subtle.\footnote{The main subtelty arises from the following fact: a simple integration by parts argument shows that forms in the kernel of the Laplacian may fail to be closed and co-closed in the presence of a boundary. By the Hodge--Morrey--Friedrichs decomposition, every relative $p$-class is represented by a unique closed and co-closed $p$-form satisfying Dirichlet boundary conditions; however, the star operator {\it swaps} Dirichlet and Neumann boundary conditions, sending relative classes to absolute ones. One can see \cite{schwarz2006hodge} for more details.}
We follow \cite{atiyah1975spectralI} by choosing a metric for which the boundary is locally isometric to a cylinder (note this preserves the diffeomorphism type of $M$) and extending a form with vanishing signature on $M$ across $N \times \R^+$. In this case, the space of such {\it square-integrable} forms which are harmonic can be identified with the image of the map $H^*(M, N; \R) \to H^*(M; \R)$, or equivalently, the kernel of $H^*(M; \R) \to H^*(\partial M; \R)$.
We denote that space as $H'^*(M; \R)$.
There is a natural Hodge star operator in $H'^*(M; \R)$, generalizing the Hodge star operator in the compact case. Once again, this allow us to split $H'^2(M; \R)$ into positive and negative definite subspaces, just as above.

\subsection{From simply-connected Riemannian four-manifolds to Narain boundaries}\label{subsec:four-manifold-to-Narain}
Let $M$ be a simply-connected spin four-manifold with boundary $N$.\footnote{For example, the trace of some even Kirby diagram (see \cref{subsec:kirby-calculus}).} Note that $\Lambda = H_2(M; \Z)$ is a lattice and 
$\Lambda^* \simeq H_2(M, N; \Z)$ by \Poincare duality. Additionally, 
the intersection form $\pitchfork$ is an even pairing on $\Lambda$. 
For simplicity, let us assume that $\pitchfork$ is non-degenerate, in which case 
the map 
\begin{equation}
    f \colon \Lambda \to H_2(M; \Z) \to H_2(M, N; \Z) \simeq \Lambda^*
\end{equation}
is injective and induces an equivalence on real cohomology.
By \cref{subsec:Kirby-to-CS}, $(H_2(M; \Z), \pitchfork)$ defines a toroidal Chern--Simons theory $\CS_\pitchfork$ whose quantum theory is the topological order $\TO_N$. 

Now we fix a Riemannian metric $g$ on $M$ for which the boundary is locally isometric to a cylinder. 
Since $f \otimes \R \colon H_2(M; \R) \to H_2(M,N;\R)$ is an isomorphism, we have $H'^2(M;\R) \simeq H_2(M,N;\R) = H_2(M;\R)$.
By \cref{subsec:Hodge-theory}, the Hodge star operator gives a polarization $m_g$ on $(H^2(M; \Z), \pitchfork)$. Using this, we can define a Narain CFT $\cB_{H^2(M; \Z), \pitchfork, m_g}$ at the boundary of $\CS_{\pitchfork} = \TO_N$. Therefore we have a commutative diagram:
\begin{equation}
    \begin{tikzcd}
            \textrm{Simply-connected four-manifolds}
        \ar[rrr, "\textrm{Taking boundary}"] \ar[d, "{\cB_{H_2(M; \Z), \pitchfork, m_g}}", swap] &&& \textrm{Spin three-manifolds} 
        \ar[d, "{\TO_N}"]
        \\ 
        \textrm{Narain boundary CFTs} \ar[rrr, "\textrm{Boundary of}"] &&& 
        \textrm{Topological order}
    \end{tikzcd}
\end{equation}
In particular, when $M$ is closed  we get an ordinary Narain CFT.
\begin{example}\label{ex:S2-times-S2-and-Narain}
    Let $M = S^2 \times S^2$ together with the metric $g_{r_1} \times g_{r_2}$ given in \cref{ex:S2-times-S2}. It follows that the Narain boundary associated to $(M, g_{r_1} \times g_{r_2})$ is the scalar boson CFT at radius $R = r_1/r_2$ (\cref{ex:Narain-CFTs}).\footnote{In general we expect that every central charge $c$ ordinary Narain CFT comes from $M = (S^2 \times S^2)^{\#c}$ equipped with a specific metric $g$.}
\end{example}

\begin{example}[Lens space and chiral bosons]\label{ex:U1-level-n-3}

    Let $n$ be an even positive integer.
    Recall the even Kirby diagram
     in \cref{ex:U1-level-n-2} that produces the lens space 
     $L(n, 1)$. Since the $K$-matrix $(n)$ is positive definite, we have a unique polarization. 
    By \cref{ex:chiral-bosons}, for any Riemannian cylindrical metric on the trace of $L(n,1)$, the associated Narain boundary 
    is the chiral boson CFT at radius $n$. The same holds for $n$ an even negative integer.
\end{example}

\begin{example}[Toric code and its gapless boundaries]
    Consider the even Kirby diagram from \cref{ex:toric-code-2} whose topological order is the toric code. Recall that the $K$-matrix is 
    \begin{equation}\label{eq:toric-code-k-matrix}
        \begin{pmatrix}
            0 & 2 \\ 
            2 & 0
        \end{pmatrix}.
    \end{equation}
    Referring to \cref{footnote:polarization-and-signature}, the space of polarizations on $(\Z^2, K)$ is also $\R_+$. Therefore we also have an $\R_+$-family of gapless boundaries of the toric code.\footnote{Note that the Ising boundary of toric code is \emph{not} a member of the Narain boundaries.}  See \cite[\S 2.4]{GeneralizedNarain} for more details. The same discussion holds for $\Z_n$ gauge theories for any $n$.
\end{example}

\section{Geometry of anyon condensations}\label{sec:geo-of-anyon-cond}
In this section we turn to a particular surgery procedure that corresponds anyon condensations.

\subsection{Anyon condensation and gapped boundaries}
\label{subsec:anyon-cond-and-gapped-boundaries}
Let us quickly review the theory of anyon condensations.\footnote{We will focus on the abelian case, where the theory is significantly simpler. We refer the reader to \cite{kong2014anyon} for the theory of anyon condensation for nonabelian topological orders.}
We begin by condensing a single boson. Take a topological order $(D, q)$. Recall that an anyon $a \in D$ is a \emph{boson} if $q(a) = 0$.  This also implies that $L(a,a) = 2q(a) =0$. When we condense $a$, we enter a different topological phase where $a$ becomes identified with the vacuum groundstate. Furthermore, the anyons which braid non-trivially with $a$ get screened out, leaving behind only those anyons $b$ such that $L(a,b) = 0$, i.e. having trivial mutual braiding.
Algebraically the condensed phase is described by the topological order $(D' = \Ann_a(D)/a, q')$, where $q'([b]) = q(b)$ for $[b]$ representing a class in the quotient $D'$. It is straightforward to check that $q'$ is a well-defined quadratic form when $a$ is a boson.
 
There is a gapped domain wall $\cD$ between the original phase $(D, q)$ and the condensed phase $(\Ann_a(D)/a, q')$, whose anyon content consists of $D/a$.\footnote{Note that there is no braiding between anyons on the domain wall.} When we bring an anyon $b \in D$ to the boundary, it is sent to the equivalence class $[b] \in D/a$. Only those anyons that braid trivially with $a$ can enter the bulk of the condensed phase.

More generally, we can simultaneously condense many anyons at once: let $\cA$ be a subgroup of $D$. We say that $\cA$ consists of mutual bosons if $q(a) = 0$ for all $a \in \cA$. This implies anyons in $\cA$ have trivial mutual braidings by \eqref{eq:S-from-theta}.
The condensed phase is the topological order given by $(\Ann_{\cA}(D)/\cA, q')$, and the gapped domain wall has anyons $D/\cA$. Once again $q'$ descends from $q$. Note that $|\Ann_{\cA}(D)| = |D|/|\cA|^2$.

When $|\cA| = \sqrt{|D|}$, we say that $\cA$ is a \emph{Lagrangian subgroup} of $D$. When $\cA$ is Lagrangian the condensed phase  is trivial as $\Ann_{\cA}(D)/\cA = 0$. In this case the gapped domain wall $\cB_{\cA}$ between $(D, q)$ and the trivial phase is in fact a gapped boundary of $(D,q)$. 
The twisted partition function for the boundary $\cB_{\cA}$ is
\begin{equation}
    Z_{\cB_{\cA}}^a(\tau) = \delta_{a \in \cA}.
\end{equation} 
Note that $Z_{\cB_{\cA}}^a$ is independent of $\tau$ as it is topological. It is straightforward to see $Z_{\cB_{\cA}}^a$ satisfies modular covariance \eqref{eq:modular-covariance} if and only if $\cA$ is a Lagrangian subgroup of $D$. 

\begin{example}[Gapped boundaries of toric code]\label{ex:toric-code-3}
Consider the toric code \cref{ex:toric-code-1}. It has two Lagrangian subgroups $\cA_{e} = \{1, e\}$ and $\cA_{m} = \{1, m\}$. They are called the \emph{electric} and \emph{magnetic} gapped boundaries respectively.    
\end{example}

\subsection{Wall's surgery}\label{subsec:wall-surgery}
In  \cite{wall1962killing} Wall 
studies the effect of surgery on middle-dimensional homology in compact connected manifolds, especially with respect to killing certain homology classes. Wall gives a prescription for how to reduce a manifold of dimension $2m+1$ to an $(m-1)$-connected manifold by surgering classes in middle-dimensional homology, when $m > 1$. We adapt Wall's reasoning to the setting of three-manifolds and perform a surgery whose effect on $\HNtor$ is precisely anyon condensation.

Fix a spin three-manifold $N$. Suppose we have a boson $a \in \HNtor$ in the corresponding topological order $\TO_N = (\HNtor, q)$. We now describe a surgery procedure that produces a spin three-manifold $N'$ whose topological order $\TO_{N'}$ represents the condensed phase $(\Ann_a(\HNtor)/a, q')$.
Let $\cK$ be a knot in $N$ that represents $a$ and $\nu$ be a tubular neighborhood of $\cK$. To perform surgery we have to be choose a parallel of $\cK$. Recall that the isotopy class of the parallel of $K$ is equivalent to the preimage of $a$ under $H_1(N - \nu) \to H_1(N)$. 
By \cref{cor:torsion-element} and the fact that $L(a,a) = 2q(a) = 0$, we see that there exists a unique torsion element $\tilde{a}$ lifting $a$.
Let $\ell$ be a parallel representing $\tilde{a}$ and $N_\ell$ be the three-manifold resulting from integral Dehn surgery along $\ell$.
$N_\ell$ has a canonical spin structure by \cref{lem:lem-2} and \cref{lem:surgery-and-spin-structure}.

It remains to check that the topological order associated to $N_\ell$ is indeed the condensed phase.
\begin{proposition}\label{prop:Nprime-is-the-condensed-phase}
    We have 
    \begin{equation}
        H_1(N_\ell)_\tor = \Ann_a(H_1(N))/a
    \end{equation}
    and the quadratic form $q'$ on $N_\ell$ is induced from $q$ on $H_1(N)$.
\end{proposition}
\begin{proof}
    By \cref{lem:surgery-on-homology}, we see that $H_1(N_\ell) = H_1(N - \nu)/(\tilde{a})$. Furthermore, since $\tilde{a}$ is torsion, we have 
    \begin{equation}
        \begin{aligned}
            H_1(N_\ell)_\tor &= H_1(N - \nu)_\tor/(\tilde{a}) \\ 
            &= \Ann_a(H_1(N)_\tor)/a
        \end{aligned}
    \end{equation}
    where the second equality uses  \cref{cor:torsion-element}.
    The description of $q'$ follows directly from the construction.
\end{proof}

We also have a nice geometric interpretation of the gapped domain wall. Let $M$ be the trace of the surgery. By \cref{lem:surgery-on-homology} we see that $H_1(M)_\tor = \HNtor/a$ indeed corresponds to the anyons at the gapped boundary. The canonical maps 
\begin{equation}
    \begin{tikzcd}
       H_1(N)  \ar[dr] & & H_1(N_\ell)  \ar[dl] \\ 
        & H_1(M)
    \end{tikzcd}
\end{equation}
correspond to taking a bulk anyon (on both sides) to the domain wall.

More generally, given a subgroup $\cA$ of mutual bosons in $\HNtor$, we can iteratively perform this surgery procedure to arrive at a spin-three manifold $N'$ whose associated topological order $\TO_{N'}$ is the phase obtained by condensing $\cA$. 
Note that when $\cA$ is a Lagrangian subgroup, 
then $H_1(N')_\tor$ is $0$ as it corresponds to the trivial topological order.\footnote{In this case $H_1(N')$ is torsion-free but not necessarily zero.}

\begin{example}[Toric code part III]\label{ex:toric-code-4}
    Let $N$ be the spin three-manifold constructed in \cref{ex:toric-code-2} using the Kirby diagram 
    
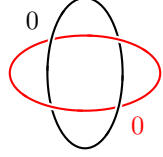
\begin{equation}
        \begin{tikzpicture}
        \begin{knot}[
        flip crossing=1,
        flip crossing=3,
        ]
        \strand[thick] (0,0) ellipse (0.5cm and 1cm);
        \strand[red, thick] (0,0) ellipse (1cm and 0.5cm);
        \node at (-0.7, 0.7) {$0$};
        \node[red] at (0.7, -0.7) {$0$};
        \end{knot}
        \end{tikzpicture}     
    \end{equation}
    Let us describe the Wall surgery associated to the electric boundary $\cA_{e} = (1, e)$ defined in \cref{ex:toric-code-3}. 
    Recall that the $e$ anyon corresponds to the meridian $m$ around the red unknot, with the $0$-framing on $m$ giving the desired parallel class with torsion homology. The resulting three-manifold $N'$ corresponds to the Kirby diagram
    \begin{equation}
        \begin{tikzpicture}
            \begin{knot}[
            flip crossing=1,
            flip crossing=3,
            flip crossing=5,
            ]
            \strand[thick] (0,0) ellipse (0.5cm and 1cm);
            \strand[red, thick] (0,0) ellipse (1.5cm and 0.5cm);
            \strand[blue, thick] (1.5, 0) ellipse(0.5cm and 0.5cm);
            \node at (-0.7, 0.7) {$0$};
            \node[red] at (0.7, -0.7) {$0$};
            \node[blue] at (2.2, 0) {$0$};
            \end{knot}
            \end{tikzpicture}   
    \end{equation}
    Applying the slam dunk move \cite[Figure 5.30]{gompf} to the blue meridian, the Kirby diagram above is equivalent to 
    
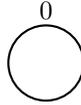
\begin{equation}
        \begin{tikzpicture}
            \draw[thick] (0,0) circle (0.5);
            \node at (0, 0.7) {$0$};
          \end{tikzpicture}
        \end{equation}
    It follows that $N' = S^1 \times S^3$ with $H_1(N')_\tor = \Z_\tor = 0$. The topological order $\TO_{N'}$ is indeed the trivial phase. The same discussion also holds for the Wall surgery associated to the magnetic Lagrangin subgroup $(1,m)$.
\end{example}

\section{From bordisms to domain walls}
\label{sec:general-bordism}

Let $N$ and $N'$ be spin three-manifolds and $M$ a spin bordism between them together with a Riemannian metric that is cylindrical at the boundaries. In this section we carry out an analysis of homology groups which allows us to assign a domain wall $\cD_M$ to the interface between $\TO_N$ and $\TO_{N'}$. The  domain wall $\cD_M$ is a combination of a gapped domain wall obtained by anyon condensation together with a Narain boundary, subsuming the discussions in \cref{sec:geometry-of-gapless-boundaries} and \cref{sec:geo-of-anyon-cond}.

\subsection{A review on \Poincare duality}
\label{subsec:Poincare}
First we will collect some basic properties of \Poincare--Lefschetz duality for four-manifolds with boundaries.\footnote{
See \cite[\S 3]{brumfiel1973quadratic} for a more geometric approach to our discussion.}
Let $M$ be a compact four-manifold and $N = \partial M$ be its boundary.
Consider the long exact sequence of relative homology groups:
\begin{equation}\label{eq:LES-of-homology}
    \cdots \to H_2(M) \xrightarrow{f} H_2(M, N) \xrightarrow{\partial} H_1(N) \xrightarrow{g} H_1(M) \to \cdots 
\end{equation}
Equip $M$ with an orientation. \Poincare duality gives us the following objects:
\begin{enumerate}
    \item an intersection form 
    \begin{equation}
      \pitchforkrel \colon H_2(M) \otimes H_2(M,N) \to \Z,
    \end{equation}
    \item a symmetric non-degenerate pairing 
    \begin{equation}
        L \colon \HNtor \otimes \HNtor \to \QmodZ,
    \end{equation}
    obtained from the induced orientation on $N$ (defined in \eqref{eq:L-N}), and
    \item a non-degenerate pairing
    \begin{equation}
        L' \colon H_2(M,N)_\tor \otimes H_1(M)_\tor \to \QmodZ
    \end{equation}
    defined similarly to $L$.
        \setcounter{listcounter}{\value{enumi}}
\end{enumerate}
Furthermore, they satisfy the following  conditions:
\begin{enumerate}\label{enum:1}
\setcounter{enumi}{\value{listcounter}}
    \item \label{item-enum:1}
    $\pitchforkrel$ induces a perfect pairing between $H_2(M)^\free$ and $H_2(M,N)^\free$.
    \item The composite 
    \begin{equation}
       \pitchfork \colon H_2(M) \otimes H_2(M) \xrightarrow{\id \otimes f} H_2(M) \otimes H_2(M,N) \xrightarrow{\pitchforkrel} \Z
    \end{equation}
    is a symmetric pairing. Indeed this is the intersection form on closed surfaces in $M$.
    \item Given $a, b \in \HNtor$ and lifts $\tilde{a}, \tilde{b} \in H_2(M,N)$, suppose that $ra = 0$ and hence $r\tilde{a} = f\tilde{a}'$ for some $\tilde{a}' \in H_2(M)$. Then 
    \begin{equation}
        L(a,b) = \frac{\tilde{a}' \pitchforkrel \tilde{b}}{r} \mod \Z.
    \end{equation}
    \item Given $a, b \in \HNtor$ and a lift $\tilde{a} \in H_2(M,N)$ of $a$,
    \begin{equation}\label{eq:L-L'-comp}
        L(a,b) = L'(\tilde{a}, gb).
    \end{equation}
        \setcounter{listcounter2}{\value{enumi}}    
\end{enumerate}
Now we equip $M$ with a spin structure.
Then we have the following:
\begin{enumerate}
\setcounter{enumi}{\value{listcounter2}}
    \item The intersection form $\pitchfork$ is an \emph{even} bilinear form.
    \item There exists a quadratic form $q \colon \HNtor \to \QmodZ$ obtained from the induced spin structure on $N$ (defined in \eqref{eq:definition-of-q}).
    \item Given $a \in \HNtor$ with lift $\tilde{a} \in H_2(M,N)$, suppose that $ra = 0$ and hence $r\tilde{a} = f\tilde{a}'$ for some $\tilde{a}' \in H_2(M)$. Then 
    \begin{equation}\label{eq:q-pitchfork-comp}
        q(a) = \frac{\tilde{a}'\pitchforkrel \tilde{a}}{2r} = \frac{\tilde{a}' \pitchfork \tilde{a}'}{2r^2}.
    \end{equation}
\end{enumerate}

\subsection{The case when $N'$ is empty}\label{subsec:N'-is-empty}
Now we return to constructing $\cD_M$.
Let us first consider the situation that $N' = \varnothing$; in that case, $M$ is a spin four-manifold bounding $N$ and $\cD_M$ will be a boundary of $\TO_N$.

Recall the long exact sequence \eqref{eq:LES-of-homology}.
Let $\cA$ be the image of $H_2(M, N)_\tor$ in $H_1(N)_\tor$ under $\partial$. 
It follows from \eqref{eq:q-pitchfork-comp} that 
$\cA$ is a subgroup of mutual bosons, i.e. $q(a) = 0$ for all $a\in \cA$.
By \cref{subsec:anyon-cond-and-gapped-boundaries}, we can condense the bosons in $\cA$ and get a topological order 
$\TO_{N/\cA}$ whose anyons are given by $\Ann_{\cA}(H_1(N)_\tor)/\cA$. We denote by $\cD_{\cA}$ the gapped domain wall between $\TO_N$ and $\TO_{N/\cA}$.

By \eqref{eq:L-L'-comp}, we have the following lemma:
\begin{lemma}\label{lem:ann-cA}
    The subgroup $\Ann_{\cA}(H_1(N)_\tor)$ is the kernel of 
    $g \colon H_1(N)_\tor \to H_1(M)$.
\end{lemma}

Now we would like to construct a Narain boundary CFT to $\TO_{N/\cA}$. 
To do that we need to find a lattice $\Lambda$ together with an even non-degenerate pairing $K$ whose discriminant group is $\Ann_{\cA}(H_1(N)_\tor)/\cA$.
Note that the even symmetric pairing $\pitchfork$ on $H_2(M)$ is not necessarily non-degenerate. However, $\pitchfork$ induces a non-degerate pairing $\pitchfork'$ on $H'_2(M) \coloneqq H_2(M)/(\ker \, \pitchfork)$. We will take $\Lambda$ to be $H'_2(M)$ and $K$ to be $\pitchfork'$. Furthermore, let 
\begin{equation}
  H'_2(M,N) = \left(\partial^{-1}(\HNtor)\right)^{\free}  = \partial^{-1}(\HNtor)/H_2(M,N)_\tor.
\end{equation}
Now the duality pairing $\pitchforkrel$ between $H_2(M)^\free$ and $H_2(M,N)^\free$ in Item \eqref{item-enum:1} descends to a duality pairing between $H'_2(M)$ and $H'_2(M,N)$. Therefore $H'_2(M,N) \simeq H'_2(M)^*$.
The long exact sequence \eqref{eq:LES-of-homology} induces the following sequence:
    \begin{equation}\label{eq:this-sequence}
    0 \to H'_2(M)\to
    H'_2(M,N)
    \to \Ann_{\cA}(H_1(N)_\tor)/\cA \to 0.
    \end{equation}
\begin{lemma}\label{lem:SES}
   \eqref{eq:this-sequence} is short exact. 
\end{lemma}
\begin{proof}
    The injectivity of the first map is clear as it is the adjoint to the non-degenerate pairing $\pitchfork'$.
    The second map is surjective by \cref{lem:ann-cA}. The exactness follows from the exactness of \eqref{eq:LES-of-homology} at $H_2(M,N)$.
\end{proof}
By \cref{subsec:Hodge-theory} the Riemannian metric on $M$ defines a polarization $m$ on  $H'_2(M) \otimes \R \simeq H'^2(M; \R)$. 
Following \cref{subsec:four-manifold-to-Narain}, we have a Narain boundary 
$\cB_{H'_2(M), \pitchfork', m}$ to $\TO_{N/\cA}$. 
At last, we define the boundary $\cD_M$
to be the composite
\begin{equation}\label{eq:def-of-cD}
    \cD_M \coloneqq 
    {\TO_{N}} \xrightarrow{\cD_{\cA}} \TO_{N/\cA} \xrightarrow{\cB_{H'_2(M), \pitchfork', m}} \mathbbm{1},
\end{equation}
where $\mathbbm{1}$ is the trivial phase.

The twisted partition function of $\cD_M$ can be written as follows:
\begin{equation}\label{eq:twisted-partition-function}
    Z_{\cD_M}^a(\tau) = |\cA|\, \delta_{a \in \Ann_{\cA}(H_1(N)_\tor)} \,
    \sum_{[\tilde{a}] \mapsto [a]}
    e^{\pi i y H([\tilde{a}]) + \pi i x K([\tilde{a}])}.
\end{equation}
Here $[a]$ represents the equivalence class of $a$ in $\Ann_{\cA}(H_1(N)_\tor)/\cA$ and $[\tilde{a}]$ sums over elements in $H'_2(M,N)$ that map to $[a]$. 
Note that $H$ and $K$ are defined using the  $\pitchfork'$ and  $m$ (see \eqref{eq:def-of-K} and \eqref{eq:def-of-H}).

It is clear that \eqref{eq:twisted-partition-function} is a mix of anyon condensation and Narain boundary CFT:
the factor $$|\cA|\, \delta_{a \in \Ann_{\cA}(H_1(N)_\tor)}$$ comes from the gapped domain wall $\cD_{\cA}$ between $\TO_N$ and $\TO_{N/\cA}$, while $$e^{\pi i y H([\tilde{a}]) + \pi i x K([\tilde{a}])}$$ 
is the partition function of the Narain boundary $\cB_{H'_2(M), \pitchfork', m}$ for $\TO_{N/\cA}$.

\subsection{The general case of bordisms between spin three-manifolds}\label{subsec:folding-trick}

Now we return to the general case where $M$ is a bordism from $N$ to some non-empty $N'$. We will construct $\cD_M$ as a domain wall between $\TO_N$ and $\TO_{N'}$. On the topological side, we can view $M$ as a bordism from $\partial M = N \sqcup -N'$ to $\varnothing$, where $-N'$ is $N'$ equipped with the opposite spin structure. Using the construction above we get a boundary for 
\begin{equation}
    \TO_{N \sqcup -N'} = \TO_{N} \boxtimes \TO_{-N'} = \TO_{N} \boxtimes \overline{\TO_{N'}}.
\end{equation}
Here $\boxtimes$ is the \emph{Deligne tensor product} and the conjugate $\overline{\TO_{N'}}$ is the parity-reversed topological order. Note that in the abelian case, $(D_1, q_1) \boxtimes (D_2, q_2) = (D_1 \oplus D_2, q_1 \oplus q_2)$ and $\overline{(D, q)} = (D, -q)$ (see \cite{belov2005classification}).
Now we use the folding trick (\cite[\S 7]{kapustin2011topological}) in the theory of topological orders, which says that a domain wall between two topological orders $\cA$ and $\cB$ is equivalent to a boundary for $\cA \boxtimes \overline{\cB}$.\footnote{Mathematically this follows from the dualizability of modular tensor categories.} 
Finally, we see that the boundary $\cD_M$ of 
$\TO_{N} \boxtimes \overline{\TO_{N'}}$ unfolds to a domain wall $\TO_{N}$ and $\TO_{N'}$.

We conclude this section by noting that the construction $M \mapsto \cD_M$ is \emph{not} functorial. The reader can readily verify this by considering the composition of the bordism $D^2 \times S^2 \colon \varnothing \to S^1 \times S^2$ with its opposite.
The fundamental reason for this is because we only consider elements in $H_2(M, N)$ that map to torsion elements in $H_1(M)$, and this restriction is not preserved by the composition of bordisms. 
In a forthcoming paper we provide a functorial construction in the setting of twisted supersymmetric theories.

\section{Outlook}\label{sec:outlook}

In this paper we investigated the correspondence between spin three-manifolds and bosonic abelian topological orders. 
Interestingly, it seems like there is a similar correspondence between $\Spinc$ three-manifolds and \emph{fermionic} topological orders. The starting point of this correspondence is the observation that a $\Spinc$ structure $\sigma^c$ on a three-manifold $N$ induces a quadratic function $q^c$ on $\HNtor$ \cite{deloup}, generalizing the quadratic form associated to a spin structure. On the other hand, the pair $(\HNtor, q^c)$ defines a fermionic topological order \cite{belov2005classification}. We expect much of the correspondence generalizes to the $\Spinc$ case. For example, every Kirby diagram should give a fermionic toroidal Chern--Simons theory, and every $\Spinc$ bordism should give a domain wall between the fermionic topological orders.

This correspondence between spin three-manifolds and three-dimensional topological orders also fits nicely into the paradigm of 3D-3D correspondences 
\cite{MR3148093, gadde2016fivebranes, MR3584441,gukov2017fivebranes}. In particular, the topological orders $\TO_N$ should be the compactification of the self-dual 6D theory  $Z^{6d}$ \cite{verlinde1995global} on the three-manifold $N$. 
This would naturally explain why bordisms give rise to domain walls, for instance. 
We look forward to extending our explicit descriptions to the non-abelian case, and relating our results to the literature on 3D-3D correspondences, in particular the aspects relating more to topological order \cite{MR4541920,bonetti2024mtc}.

\bibliographystyle{amsalpha}
\bibliography{bib.bib}

\end{document}